\def\MINE{0}

\if\MINE1
    \documentclass[11]{article}
    \usepackage{amsthm}
\else
    \documentclass[runningheads]{llncs}
    \usepackage[T1]{fontenc}
\fi

\usepackage{marginnote} %

\usepackage{ifthen}
\usepackage{latexsym,amssymb,amsmath,mathdots}
\usepackage{graphics}
\usepackage[dvipsnames]{xcolor}
\usepackage{graphicx}
\usepackage{tcolorbox}
\usepackage[colorlinks,citecolor=blue]{hyperref}
\usepackage{bm}
\usepackage{shuffle}   %
\usepackage{xspace}
\usepackage{tikz}     
\usepackage{tkz-graph}  
\usetikzlibrary{%
  arrows,%
  positioning,%
  decorations.pathmorphing,%
  decorations.pathreplacing,%
  automata,%
  shapes.geometric,%
}

\if\MINE1
\newtheorem{theorem}{Theorem}[section]

\newtheorem{definition}[theorem]{Definition}

\theoremstyle{definition}

\newtheorem{example}[theorem]{Example}

\fi

\newcommand\pnsi{\par\indent}
\newcommand\pnsn{\par\noindent}
\newcommand\pssi{\par\smallskip\indent}
\newcommand\pssn{\par\smallskip\noindent}

\newcommand\pmsn{\par\medskip\noindent}

\newcommand\pbsn{\par\bigskip\noindent}

\newcommand\emdef[1]{\textsf{#1}}
\newcommand\emshort[1]{\emph{#1}}

\newenvironment{fsm}[1][]
	{\begin{center}
	\begin{tikzpicture}[font=\footnotesize,shorten >=1pt,node distance=3.5cm,
	on grid,>=stealth',initial text=,
	every node/.style={align=center},
	every state/.append style={inner sep=1pt},
	every path/.style={->,bend angle=20},
	#1]}
	{\end{tikzpicture}\end{center}}

\newcommand{\N}{\ensuremath{\mathbb{N}}\xspace}

\newcommand{\al}[0]{\ensuremath{\Sigma}\xspace} 
\newcommand{\alstar}[0]{\ensuremath{\al^*}\xspace} 
 
\newcommand{\alG}[0]{\ensuremath{\Gamma}\xspace}
\newcommand{\ew}[0]{\ensuremath{\lambda}\xspace}   
 
\newcommand\rel{\ensuremath{\omega}\xspace}   
\newcommand\po{\ensuremath{\le}\xspace}   
\newcommand\ppx{\ensuremath{\po_{\mathrm{px}}}\xspace}   
\newcommand\psx{\ensuremath{\po_{\mathrm{sx}}}\xspace}   %

\newcommand\spo{\ensuremath{<}\xspace}   
\newcommand\spx{\ensuremath{\spo_{\mathrm{px}}}\xspace}   
\newcommand\ssx{\ensuremath{\spo_{\mathrm{sx}}}\xspace}   %

\newcommand{\rea}{\ensuremath{\alpha}\xspace}  

\newcommand\auta{\ensuremath{\bm{a}}\xspace}   
\newcommand{\tr}{\ensuremath{\bm{t}}\xspace}       
\newcommand{\trt}{\ensuremath{\bm{t}}\xspace}    
\newcommand\trs{\ensuremath{\bm{s}}\xspace}      
\newcommand\tsx{\ensuremath{\bm t_{\rm sx}}\xspace}   
\newcommand\thc{\ensuremath{\bm t_{\rm hc}}\xspace}   
\newcommand\ted{\ensuremath{\bm t_{\text{1-ed}}}\xspace}   

\newcommand\indepfont{\mathbb}

\newcommand\descrf{\ensuremath{\varphi}\xspace}  
\newcommand\indepf{\ensuremath{\indepfont{P}_{\descrf}}\xspace}
\newcommand{\indep}[1]{\ensuremath{\indepfont{P}_{#1}}\xspace}

\newcommand{\indephyph}[2]{\ensuremath{\indepfont{P}_{#1\text{-}#2}}\xspace}

\newcommand\indepP{\ensuremath{\indepfont{P}}\xspace}

\begin{document}

\newcommand{\thetitle}{Some Remarks on Marginal Code Languages}

\newcommand{\theabstract}{A prefix code L satisfies the condition that no word of L is a proper prefix of another word of L. Recently, Ko, Han and Salomaa relaxed this condition by allowing a word of L to be a proper prefix of at most k words of L, for some `margin' k, introducing thus the class of k-prefix-free languages, as well as the  similar classes of k-suffix-free and k-infix-free languages. 
Here we unify the definitions of these three classes of languages into one uniform definition in
 two ways: via the method of partial orders and via the method of transducers. 
 Thus, for any known class of code-related languages definable via the transducer method, one gets a marginal version of that class. 
 Building on the techniques of Ko, Han and Salomaa, we discuss the \emph{uniform} satisfaction and maximality problems for   marginal classes of languages.}

\if\MINE1
\begin{center}
	\textbf{\Large\thetitle}
\pbsn
{\large Stavros Konstantinidis}  
\pbsn
\parbox{0.90\textwidth}{Department of Mathematics and Computing Science, Saint Mary's University, Halifax, NS, Canada, \texttt{s.konstantinidis@smu.ca}}
\pbsn
\parbox{0.90\textwidth}{
\textbf{Abstract.}  \theabstract
}  
\end{center}
\else   
\title{\thetitle}
\titlerunning{Remarks on Marginal Code Languages}
\author{Stavros Konstantinidis\orcidID{0000-0002-6628-067X}}
\authorrunning{S. Konstantinidis}

\institute{
Mathematics and Computing Science,\\
Saint Mary's University, Halifax, NS, Canada,\\
\email{s.konstantinidis@smu.ca}\\
\url{https://cs.smu.ca/~stavros/}
}
\maketitle
\begin{abstract}
\theabstract
\end{abstract}
\fi      

\section{Introduction}\label{sec:intro}
The concepts of $k$-prefix-free and finitely prefix-free languages defined in \cite{KoHanSalomaa:2021} extend naturally the concept of prefix code, and similarly the corresponding concepts for suffix and infix codes, \cite{KoHanSalomaa:2021}. 
These three concepts can be unified into a single uniform concept of a $k$-$\alpha$ and finitely-$\alpha$ language that extends naturally any existing concept of an $\alpha$ language which is definable in the method of partial orders \cite{Shyr:book}, or the method of transducers \cite{DudKon:2012,Kon:2017}.
The generic term for codes like prefix and suffix codes is \emdef{independent languages}.
Following the use of the term `margin' in \cite{KoHanSalomaa:2021} for the above parameter $k$, here we use the generic term \emdef{marginal independent languages} for the $k$-$\alpha$ and finitely-$\alpha$ versions of existing independent languages. 
Building on the techniques of \cite{KoHanSalomaa:2021}, we discuss the uniform satisfaction and maximality problems for  marginal classes of languages.\pssn
\textbf{Organization and main results.}
The paper is organized as follows. The next section contains a few basic concepts from automata, formal languages, and codes. 
In \underline{Section~\ref{sec:jurg}}, we review two \emph{methods} of defining classes of independent languages---these classes are called \emdef{independent properties:} the method of partial orders and the more  broad method of J\"urgensen independence. 
We discuss how to extend the method of partial orders to a method for `marginal' partial orders, in the spirit of \cite{KoHanSalomaa:2021}. 
In Theorem~\ref{th:alpha:margin}, we show the following mathematically interesting result: every language satisfying a finitely-$\alpha$ independent property definable via a right infinite partial order cannot be embedded into a maximal language and, therefore, these  finitely-$\alpha$ independent properties fall outside the method of J\"urgensen independence.
In \underline{Section~\ref{sec:dec}}, we review the \emdef{satisfaction problem}, that is, the problem of deciding whether a given regular language satisfies one of the three types of marginal properties in \cite{KoHanSalomaa:2021}. 
In \underline{Section~\ref{sec:formal}}, we review the `formal' method of transducers for defining independent properties, where the property is describable via a finite expression representing a transducer. 
This leads to uniform versions of the satisfaction and maximality decision problems, where the independent property is part of the input. 
In Theorems~\ref{th:gen:sat:margin} and~\ref{th:gen:sat:finite:margin}, we show how to decide the \emph{uniform} satisfaction problem for marginal properties.
In \underline{Section~\ref{sec:last}}, we close with a few concluding remarks and a few directions for future research.

\section{Basic Notions and Notation}\label{sec:notation}
We assume the reader to be familiar with basics of formal languages, see e.g., \cite{HopcroftUllman,MaSa:handbook,FLhandbookI,Salomaa:1973}. Some notation:
\ew denotes the empty word;
$\al,\alG$ denote arbitrary alphabets; 
$\bar L$ denotes the complement of the language $L$.
Let  $L$ be a language  and let $u,v,w,x$ be any words. If $w\in L$ then we say that $w$ is an \emdef{$L$-word}. We use standard operations and notation on words and languages; in particular $|w|$ denotes the length of $w$. 
If $w$ is of the form $uv$ then $u$ is called a \emdef{prefix} of $w$ and $v$ is called a \emdef{suffix} of $w$. If $w$ is of the form $uxv$ then $x$ is called an \emdef{infix} of $w$ and $uv$ is called an \emdef{outfix} of $w$.
A prefix $u$ of $w$ is called \emdef{proper}, if $u\not=w$. Similar are the concepts of proper suffix, proper infix and proper outfix.
\pnsi
We use the term \emdef{(finite) automaton} for the standard notion of a nondeterministic finite automaton (NFA), and the acronym DFA for deterministic finite automaton. 
We  assume the reader to be familiar with basics of  transducers, see e.g., \cite{Be:1979,Sak:2009,Yu:handbook}.
A \emdef{(finite) transducer} is a 6-tuple $\trt=(Q,\al,\alG,E,I,F)$ 
such that $Q$ is the set of states, $I\subseteq Q$ is the nonempty set of start (initial) states, $F\subseteq Q$ is the set of final states,  $\al,\alG$ are the  input and output alphabets, respectively, and $E$ is the finite set of transitions (edges). \emph{In this paper, all transducers have equal input and output alphabets: $\al=\alG$.}
We assume that all transducers are in \emdef{standard form:} in each transition $(p,x/y,q)\in E$, we have that the \emdef{input label} $x$ is either the empty word \ew or a symbol in \al, and the same for the \emdef{output label} $y$.
We view a transducer \trt as a labelled directed graph, so we can talk about a \emdef{path} in \trt.  
The \emdef{label} of a path is the pair $(u,v)$ of words such that $u$ (resp. $v$) is the concatenation of the input (resp. output) labels on the transitions of the path.
A path is \emdef{accepting} if the first state of the path is initial and the last state is final.
The relation $R(\trt)$ \emdef{realized by} \trt is the set of labels $(u,v)$ on the accepting paths of \trt.
The set of outputs of \trt on input $u$ is denoted by $\trt(u)$, that is, $v\in\trt(u)$ iff $(u,v)\in R(\trt)$.
The \emdef{inverse $\trt^{-1}$ of \trt} realizes the inverse of the relation $R(\trt)$, that is, $R(\trt^{-1})=\{(v,u)\mid (u,v)\in R(\trt)\}$.
If \trs is an automaton, or transducer, then $|\trs|$ denotes the \emdef{size of \trs} = the sum of the numbers of states and transitions of \trs.

We write $\N,\N_0$ for the sets of positive integers and non-negative integers, respectively. If $S$ is a set then $|S|$ denotes the cardinality of $S$.
For a set $S$, the notation $[S]^n$ is used for the 
\emdef{$n$-fold cross product}  $S\times\cdots\times S$ of $S$. This notation is used to avoid confusion with the $n$-th power $L^n$ of a language  $L$.
\pnsi
A \emdef{language property} is simply a class (set) of languages.
An \emdef{independent property} is any  language property \indepP for which every language $L\in\indepP$ is included in a \indepP-maximal language $M\in\indepP$. 
By a \emdef{\indepP-maximal language} (or maximal language with respect to \indepP) we mean a language $M$ in \indepP such that $M\cup\{w\}\notin\indepP$, for all words $w\notin M$.
With very few exceptions, the various concepts of codes in the literature are captured by the concept of independence.
When $L\in \indepP$ we say that $L$ \emdef{satisfies} the property \indepP, or that $L$ is an \emdef{independent language} w.r.t. \indepP.
Usually, an independent property has a name $\varphi$ and is denoted as $\indepP_{\varphi}$, in which case a language $L$ satisfying $\indepP_{\varphi}$ is also called a \emdef{$\varphi$-independent} language.
\pnsi
As stated in \cite{JuKo:handbook}, the prime examples of independent languages are those studied in the theory of codes. 
Each class of codes is defined by a certain condition on the words of the code (the language). 
\begin{example}\label{ex:various:codes}
	We recall a few well known classes of codes.
\begin{itemize}
    \setlength{\itemsep}{0pt}%
    \setlength{\parskip}{0pt}%
    \vspace{-0.5\topsep}
	\item A language $L$ is a \emdef{prefix code} if no $L$-word is a  prefix of another $L$-word.
	\item A language $L$ is a \emdef{suffix code} if no $L$-word is a  suffix of another $L$-word.
	\item A language $L$ is an \emdef{infix code} if no $L$-word is an  infix of another $L$-word.
	\item A language $L$ is an \emdef{outfix code} if no $L$-word is an  outfix of another $L$-word.
	\item A language $L$ is a \emdef{hypercode} if no $L$-word results by deleting a number of symbols in another $L$-word.
	\item A language $L$ is a \emdef{uniquely decodable/decipherable code}, or \emdef{UD-code} for short, if every word can be written in at most one way as a concatenation of $L$-words.
\end{itemize}
\end{example}

\pnsn\textbf{Abbreviations:}\; iff = ``if and only if'',\quad
w.r.t. = ``with respect to''

\section{J\"urgensen Independence and Marginal Code Languages}\label{sec:jurg}
An important direction in the research on independent languages is the investigation of systematic methods that allow one to define and study different classes of these languages. 
Next we review some of these methods.

\pmsn
\textbf{Independent languages via partial orders.}
To our knowledge, the first method for defining code properties is the method of 
\cite{Shyr:Thierrin:relations} via certain binary relations on words. A binary relation \rel is a set of pairs of words, that is, $\rel\subseteq\al^*\times\al^*$. 
In \cite{Shyr:book}, the author presents the method via binary relations that are partial orders (in the standard definition). Let \po be a partial order on \alstar.
A language $L$ is called \emdef{\po-independent}, if
\begin{equation}\label{eq:po}
u,v\in L \>\text{ and }\> u\po v \>\text{ imply }\> u=v.
\end{equation}
The same class of languages is defined via the strict version `$<$' of `$\po$':
\begin{equation}\label{eq:spo}
	<\>\cap\>[L]^2=\emptyset.
\end{equation}
Above, $<\,=\,\po-\{(u,u)\mid u\in\alstar\}$, which is the irreflexive and asymmetric version of `$\po$'.
The classic examples of prefix and suffix codes are defined via the partial orders $\ppx, \psx\subseteq[\alstar]^2$, respectively:
\[
\ppx=\{(w,z)\mid z \text{ is a prefix of } w\} \quad\text{ and }\quad \psx=\{(w,z)\mid z \text{ is a suffix of } w\};
\]
or via the strict versions of these orders 
\pmsn\quad
$\spx=\{(w,z)\mid z \text{ is a proper prefix of } w\}$  and
\pssn\quad
$\ssx=\{(w,z)\mid z \text{ is a proper suffix of } w\}.$
\pmsn
Thus, $L$ is a prefix code iff no two \emph{different} words $u,v\in L$  are related via $\ppx$; equivalently, iff no two words $u,v\in L$  are related via $\spx$.


\pmsn
\textbf{J\"urgensen Independence.} 
Let $n$ be a positive integer or $\aleph_0$. An \emdef{$n$-independent property}  is a class of languages \indepP such that
\begin{equation}
L\in\indepP \quad\text{iff}\quad \text{for all $L'\subseteq L$ with $|L'|<n$: $L'\in\indepP$}.
\end{equation}
A \emdef{J\"urgensen property} is simply an $n$-independent property for some $n\in\N\cup\{\aleph_0\}$.
Thus, to tell whether a language $L$ is independent w.r.t. some \indepP, it is sufficient (and necessary) to ensure that every subset of $L$ with $<n$ elements is independent w.r.t. \indepP. 
For example, the class of prefix codes is a \textbf{3}-independent property because a language $L$ is a prefix code iff every subset of $L$ with at most \textbf{2} words is a prefix code.
Why is the J\"urgensen definition interesting? 
One reason is that it can be used to model the independent properties defined by all previous methods. 
Another reason is that the J\"urgensen method guarantees that every independent language is included in a maximal one \cite{JuKo:handbook}.
A third reason is that the class of UD-codes cannot be modelled by the previous methods, as it is an $\aleph_0$-independence
(a language $L$ is a UD-code iff every finite subset of $L$ is a UD-code), but the class of UD-codes is \emshort{not} an $n$ independent property, for any $n<\aleph_0$, \cite{JuKo:handbook}.

\subsection{Marginal properties.}\label{sec:marginal:def}
The \po-independent properties defined in~\eqref{eq:spo} require that no pair of $L$-words are related via the order `$<$'.
For example, when $L$ is a prefix code, no $L$-word is a proper prefix of an $L$-word.
In \cite{KoHanSalomaa:2021}, Ko, Han and Salomaa, relax this condition by allowing an $L$-word to be a proper prefix of at most $k$ words, for some `error' margin $k\in\N_0$. 
They define the class of \emdef{$k$-prefix-free languages}, which is here denoted by \indephyph{k}{\mathrm{px}}, as follows
\begin{equation}\label{eq:k:prefix}
	L\in \indephyph{k}{\mathrm{px}},\quad\text{if every word of $L$ is a proper prefix of at most $k$ words of $L$.} 
\end{equation}
It follows that \indephyph{0}{\mathrm{px}} is exactly the class of prefix codes.
The class \indephyph{\mathrm{fin}}{\mathrm{px}} of \emdef{finitely prefix-free languages} is defined as follows  
\begin{equation}\label{eq:fin:prefix}
	L\in \indephyph{\mathrm{fin}}{\mathrm{px}},\quad\text{if $\forall w\in L$, $w$ is a proper prefix of at most finitely many $L$-words.} 
\end{equation}
Unlike the case of prefix codes (0-prefix languages) which are UD-codes, it is important to note that, for $k>0$, these languages are not necessarily UD-codes. For example, $\{0,00\}$ is a 1-prefix-free language but it is not a UD-code.
\pnsi
The classes \indephyph{k}{\mathrm{sx}} and \indephyph{k}{\mathrm{ix}} of \emdef{$k$-suffix-free languages} and \emdef{$k$-infix-free languages} are defined analogously, as well as the classes \indephyph{\mathrm{fin}}{\mathrm{sx}} and \indephyph{\mathrm{fin}}{\mathrm{ix}}, \cite{KoHanSalomaa:2021}.
\pnsi
We extend the above concepts via the mechanism of partial orders.

\begin{definition}\label{def:alpha:margin}
	Let $\po_\alpha$ be a partial order on \alstar, where $\alpha$ is used as the name of the partial order. The class \indephyph{k}{\alpha} of $k$-$\alpha$ languages is defined as follows:  
	\[
	L\in \indephyph{k}{\alpha}, \quad \text{if every word of $L$ is smaller w.r.t. `$\spo_{\alpha}$'  by at most $k$ words of $L$.}
	\]
The class \indephyph{\mathrm{fin}}{\alpha} of finitely-$\alpha$ languages is defined as above by using the condition ``every word of $L$ is smaller w.r.t. `$\spo_{\alpha}$'  by at most finitely many words of $L$''.
\end{definition}

Thus, we can talk about the classes of $k$-hypercode and $k$-outfix-free  languages, based on  appropriate partial orders. 
Note that every singleton language $\{w\}$ is a $k$-$\alpha$ language, for every partial order $\po_\alpha$ and every $k\in\N_0$.
While for all $k\in\N_0$, the classes \indephyph{k}{\alpha} are independent properties, this is not the case for the classes \indephyph{\mathrm{fin}}{\alpha}. 
This is shown in the next theorem, in which the term \emdef{right-infinite order} $\po_\alpha$ means that, for every word $u$, there is a word $v$ such that $u\spo_\alpha v$.

\begin{theorem}\label{th:alpha:margin}
	Let $\po_\alpha$ be any partial order on words and let $k\in\N_0$.
	The following statements hold true.
	\begin{enumerate}
		\item The class \indephyph{k}{\alpha} is a (k+3)-independent property; hence, every $k$-$\alpha$ language is included in a maximal one.
		\item If $\po_\alpha$ is a right-infinite order then, for every finitely-$\alpha$ language $L$, there is a word $z\notin L$ such that $L\cup\{z\}$ is a finitely-$\alpha$ language. Hence, the class \indephyph{\mathrm{fin}}{\alpha} is not a J\"urgensen property.
	\end{enumerate}
\end{theorem}
	\begin{proof}
		The first statement follows from the definitions of J\"urgensen property and $k$-$\alpha$ language, using standard logical arguments. 
		For the second statement, consider any finitely-$\alpha$ language $L$. We show that there is a word $z\notin L$ such that $L\cup\{z\}$ is a finitely-$\alpha$ language.
		If $L$ is empty then $z$ can be any word: $L\cup\{z\}=\{z\}$, which is finitely-$\alpha$. If $L$ is not empty, consider any word $u\in L$. 
		We use the notation $V(x) = \{v\mid x\spo_\alpha v\}$, for any word $x$.
		As $L$ is finitely-$\alpha$, the set $L\cap V(u)$ is finite; so it is of the form 
		\[
		L\cap V(u)\>=\>\{v_1,\ldots,v_n\}.
		\]
		Consider any element $v_r$ which is maximal in $\{v_1,\ldots,v_n\}$; that is, $v_r\not\spo_\alpha v_j$ for all $j=1,\ldots, n$.
		As $\po_\alpha$ is right-infinite, there is a word $z$ such that $v_r\spo_\alpha z$. Moreover, we have that $u\spo_\alpha v_r\spo_\alpha z$ and $z\notin\{v_1,\ldots,v_n\}$.
		We show next that $z\notin L$ and $L\cup\{z\}$ is a finitely-$\alpha$ language.
		\pnsi
		As $z\in V(u)$ and $z\notin L\cap V(u)$, we have that $z\notin L$. 
		\pnsi
		Let $w$ be any word in $L\cup\{z\}$. 		
		We need to show that  $w$ is smaller w.r.t. `$\spo$'  by at most finitely many words of $L\cup\{z\}$; that is, the set $\big(L\cup\{z\}\big)\cap V(w)$ is finite. First note that
		\[
		\big(L\cup\{z\}\big)\cap V(w)\>=\> \big(L\cap V(w)\big)  \cup  \big(\{z\}\cap V(w)\big).
		\]
		If $w\not=z$ then $w\in L$ and $L\cap V(w)$ is finite and, therefore $\big(L\cup\{z\}\big)\cap V(w)$ is finite as well.
		If $w=z$ then it is sufficient to show that $L\cap V(z)$ is empty. 
		If there were $z_1\in L\cap V(z)$ then $u\spo_\alpha z \spo_\alpha z_1$, implying that $z_1\in V(u)$ and also $z_1\in\{v_1,\ldots,v_n\}$, which is impossible.
	\end{proof}
All partial orders  corresponding to the classes of codes in Example~\ref{ex:various:codes}, other than the class of UD-codes, are right-infinite orders.
\pnsi
It is interesting to note that, if a language $L$ is regular and finitely prefix-free then $L$ is $k$-prefix-free for some $k\ge0$; but this is not the case, in general, if $L$ is context-free \cite{KoHanSalomaa:2021}.

\section{Decidability of Satisfaction}\label{sec:dec}
The \emdef{satisfaction problem} for a \emph{fixed} independent property \indepP is to decide whether a given language $L$ satisfies \indepP; that is whether $L\in\indepP$. 
The language is given as input via a well-defined description method. 
Usually, the input is a finite automaton describing a regular language. 
This is a reasonable assumption in coding theory, as the languages of interest are usually finite, or regular.
Of course it is perfectly fine to assume that the input describing the language is a context-free grammar or a pushdown automaton; however in this case, the satisfaction problem is usually undecidable \cite{JuKo:handbook}.

\pnsi
In Section~\ref{sec:formal}, we discuss the \emph{uniform} versions of the above decision problems, where in addition to the description of the language $L$, the desired algorithm also takes as input a description of the  independent property~\indepP.
\pmsn
\textbf{Decidability of satisfaction of marginal properties.}
These properties are from \cite{KoHanSalomaa:2021} and are defined in the statements~\eqref{eq:k:prefix} and~\eqref{eq:fin:prefix} of   Section~\ref{sec:marginal:def}. As shown in \cite{KoHanSalomaa:2021}, the $k$-marginal properties are hard to decide---see further below for the finitely-marginal properties:
\begin{itemize}
    \setlength{\itemsep}{0pt}%
    \setlength{\parskip}{0pt}%
    \vspace{-0.5\topsep}
	\item Deciding whether the language $L(\auta)$ is $k$-prefix-free, for given automaton \auta and integer $k\ge0$, is PSPACE-complete.
	\item Deciding whether the language $L(\auta)$ is $k$-suffix-free, for given automaton \auta and integer $k\ge0$, is PSPACE-complete.
	\item Deciding whether the language $L(\auta)$ is $k$-infix-free, for given automaton \auta and integer $k\ge0$, is PSPACE-hard.
\end{itemize}
It was left as an open question in \cite{KoHanSalomaa:2021} whether the satisfaction of the $k$-infix-free property is in the class PSPACE.
Then, the question was answered in the affirmative in  \cite{IbaMcQu:2023}---in fact, even for the larger machine class where \auta is a finite-crossing 2-way machine with a number of reversal-bounded counters. 
\pnsi
For the case of the finitely marginal properties---see definition in statement~\eqref{eq:fin:prefix}---the satisfaction problem can be decided efficiently: 
\begin{itemize}
	\item Assuming a fixed (but arbitrary) alphabet, there is an $O(|\auta|^2)$ algorithm to decide whether a given automaton \auta accepts a finitely-prefix-free language (also, finitely-suffix-free, finitely-infix-free). 
\end{itemize}
We return to the satisfaction problem of marginal properties in Section~\ref{sec:marginal}.

\section{Formal Independent Properties \& Marginal Versions}\label{sec:formal}
Here we consider the idea that an independent property can be given as input to an algorithm which would return an answer to a desired question about the property. 
For this, it is necessary to have  a syntactic method of describing an  independent property via some expression \descrf.
We shall write \indepf to denote the independent property \emdef{described by} \descrf (see further below about a selection of syntactic methods for describing independent properties).
We shall use the term \emdef{$\descrf$-independent language} for a language in \indepf (a language that satisfies \indepf).
\pssi
We are interested in the \emdef{uniform satisfaction problem}: 
\begin{description}
	\item{\textsf{Input:}} the description \descrf of an independent property and the description of a language $L$
	\item{\textsf{Question:}} whether $L$ satisfies \indepf.
\end{description}
\pnsi
We are also interested in the \emdef{uniform maximality problem}: 
\begin{description}
	\item{\textsf{Input:}} the description \descrf of an independent property and the description of a language $L$
	\item{\textsf{Question:}} whether $L$ is maximal satisfying \indepf.
\end{description}

\pmsn
\textbf{The formal method of transducers.}
The method of \cite{DudKon:2012,Kon:2017,KMMR:2015} uses transducers to describe rational binary relations (including several partial orders), defining thus several 3-independent properties, as in~\eqref{eq:po}. 
Let \trt be 
a transducer.
The property $\indep{\trt}$ described by \trt is defined as follows 
\begin{equation}\label{eq:TRANSD}
	L\in\indep{\trt}, \quad\text{if}\> v\in\trt(u) \>\text{ implies }\> u=v,\>\text{ for all $u,v\in L$}.
\end{equation}
If the transducer \trt is \emdef{input-altering}, meaning that $w\notin\trt(w)$ for all words $w$, then we have a simpler characterization of \indep{\trt}:
\begin{equation}\label{eq:inp:alter}
	L\in\indep{\trt} \quad \text{iff} \quad\trt(L)\cap L=\emptyset.
\end{equation}
The transducer $\tsx$ in the figure below is input-altering and describes the class of suffix codes, that is, $L$ is a suffix code iff $\tsx(L)\cap L=\emptyset$. The transition label $a/\ew$ actually means multiple transitions, one for each $a\in\al$, and similarly for the other labels. The label $a/a'$ also means multiple transitions, one for each pair $a,a'\in\al$ with $a\not=a'$.
The transducer \thc below is input-altering and describes the class of hypercodes.
The transducer \ted  is input-altering and describes the class of 1-error detecting languages\footnote{A language $L$ is $k$-error detecting if the Hamming distance of any two different $L$-words of equal lengths is $>k$. The transducer method is expressible enough to describe error-detecting languages for many error combinations, including those in~\cite{Kon:2001}.}
\pnsn
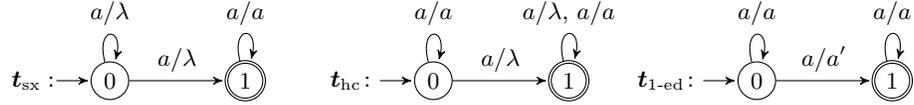
\begin{fsm}[node distance=1.8cm,every state/.style={inner sep=1pt,minimum size=0.5cm}]
	\node [state,initial] (q2) {$0$};
	\node [node distance=0.61cm,left=of q2,anchor=east] {$\tsx:$};
	\node [state,accepting,right of=q2] (q3) {$1$};
	\node [state,initial,right of=q3,node distance=2.5cm] (q4) {$0$};
	\node [node distance=0.60cm,left=of q4,anchor=east] {$\thc\colon$};
	\node [state,accepting,right of=q4] (q5) {$1$};
	\node [state,initial,right of=q5,node distance=2.5cm] (q6) {$0$};
	\node [node distance=0.60cm,left=of q6,anchor=east] {$\ted\colon$};
	\node [state,accepting,right of=q6] (q7) {$1$};
	\path 		(q2) edge node [above] {$a/\ew$} (q3)
		(q2) edge [loop above] node [above] {$a/\ew$} ()
		(q3) edge [loop above] node [above] {$a/a$} ()
        (q4) edge [loop above] node [above] {$a/a$} ()
		(q4) edge node [above] {$a/\ew$} (q5)
		(q5) edge [loop above] node [above] {$a/\ew$, $a/a$} ()	
        (q6) edge [loop above] node [above] {$a/a$} ()
		(q6) edge node [above] {$a/a'$} (q7)
		(q7) edge [loop above] node [above] {$a/a$} ()	
		;
\end{fsm}
We have that $\tsx(v)$ = set of proper suffixes of $v$. 
Notice that the transducer method realizes the inverse of the rational partial orders in Section~\ref{sec:jurg}. For example, we have that $\big({\spo_{\rm sx}}\big)^{-1}=R(\tsx)$.
\pnsi
In \cite{DudKon:2012}, it is shown that, for every regular trajectory property $\indepP_{\rea}$, there is an input-altering transducer \trt  describing the same independent property, that is, $\indepP_{\rea}=\indepP_{\trt}$.
Moreover, there are natural error-detecting properties that cannot be described by trajectories. 
The method of transducers shown above can only describe 3-independent properties. 
When the transducers involved are input-altering, the method formalizes the approach of independence via $n$-ary relations that are rational for $n=2$---see the definition at equation~\eqref{eq:omega} and note that equation~\eqref{eq:inp:alter} is equivalent to $R(\trt)\cap[L]^2=\emptyset$.
\pnsi
The use of transducers allows one to decide both, the uniform satisfaction problem and the uniform maximality problem for regular languages:

\begin{theorem}[\cite{DudKon:2012,KaKoKo:2018}]\label{th:TRANSD:sat}
	Assuming a fixed (but arbitrary) alphabet, 
	  there is\footnote{The polynomial time of the algorithm is noted in \cite{DudKon:2012}, and the $O(|\trt||\auta|^2)$ time is established in \cite[Theorem~21]{KaKoKo:2018}.} an $O(|\trt||\auta|^2)$ time algorithm to decide, given transducer \trt and automaton \auta, whether $L(\auta)\in\indep{\trt}$.
\end{theorem}

\begin{theorem}[\cite{DudKon:2012}]\label{th:TRANSD:max}
	It is PSPACE-complete to decide, given transducer\footnote{Theorem 9 of \cite{DudKon:2012} states that the problem is PSPACE-hard, but it is easy to see that it is in fact PSPACE-complete. Moreover, Theorem 9 of \cite{DudKon:2012} assumes that the given transducer \trt is input-preserving, but this restriction is not necessary---a transducer \trt is input-preserving, if $\trt(w)\not=\emptyset$ implies $w\in\trt(w)$, for all words $w$. 
	In fact, the problem is PSPACE-hard even if \trt is input-altering.} 
	\trt and automaton \auta,  whether $L(\auta)$ is maximal in \indep{\trt}.
\end{theorem}

\subsection{A formal method for marginal independent properties.}\label{sec:marginal}
Definition~\ref{def:alpha:margin} extends the marginal independent properties of \cite{KoHanSalomaa:2021} in a natural way. 
However, when we want to deal with the decidability of the \emph{uniform} satisfaction problem, we need to allow a description of  the partial order to be given as input to the desired algorithm.
We use transducers to describe many partial orders (in fact any rational binary relation).
We extend now the definition of \indep{\trt} in~\eqref{eq:TRANSD} to one for marginal properties. Note that the condition in~\eqref{eq:TRANSD} is equivalent to ``$\trt(u)\cap(L-\{u\})=\emptyset$, for all $u\in L$''. Here is the extended definition:
\begin{equation}\label{eq:TRANSD:margin}
L\in \indep{\trt,k}, \quad\text{if }\;
|\trt(u)\cap (L-\{u\})|\le k,\> \text{for all $u\in L$}.
\end{equation}
The marginal properties of \cite{KoHanSalomaa:2021} and many more can be defined as in~(\ref{eq:TRANSD:margin}) above via the use of  appropriate transducers \trt.
A motivation for introducing marginal properties is to relax the strict  emptiness condition $\trt(u)\cap (L-\{u\})=\emptyset$ to the condition $|\trt(u)\cap (L-\{u\})|\le k$, so it makes sense to assume that $k$ is small, or even fixed.
For \emph{fixed $k$}, the uniform satisfaction problem of transducer marginal properties is decidable in polynomial time for regular languages:

\begin{theorem}\label{th:gen:sat:margin}
Let $k$ be a fixed nonnegative integer. 
There is a polynomial time algorithm to decide, given  transducer \trt and automaton \auta, whether  $L(\auta)\in\indep{\trt,k}$. 
\end{theorem}
\begin{proof}
If $k=0$ then $\indep{\trt,k}=\indep{\trt}$, and the statement follows immediately from Theorem~\ref{th:TRANSD:sat}. Now let $k>0$. 
We make use of the operations $\uparrow$ and $\downarrow$ between a transducer \trs and an automaton \auta, \cite{Kon:2002,Kon:2017}: 
\begin{itemize}
    \setlength{\itemsep}{0pt}%
    \setlength{\parskip}{0pt}%
    \vspace{-0.5\topsep}
	\item $\trs\downarrow\auta$ is the transducer whose domain is $L(\auta)$ such that $v\in(\trs\downarrow\auta)(u)$ iff $v\in\trs(u)$ and $u\in L(\auta)$. The transducer $\trs\downarrow\auta$ can be computed in time $O(|\trt||\auta|)$.
	\item $\trs\uparrow\auta$ is the transducer such that $v\in(\trs\uparrow\auta)(u)$ iff $v\in\trs(u)$ and $v\in L(\auta)$. The transducer $\trs\uparrow\auta$ can be computed in time $O(|\trs||\auta|)$.
\end{itemize}
Note that $v\in(\trs\downarrow\auta\uparrow\auta)(u)$ iff $u,v\in L(\auta)$ and $v\in\trs(u)$.
We also use the following result of~\cite{GurIba:1983}: \emph{for fixed integer $k>0$, there is a polynomial time
algorithm to determine whether an arbitrary transducer is $k$-valued;}
where a transducer {\trs is $k$ valued}, if $|\trs(w)|\le k$, for all words $w$.
We would like to have that $L(\auta)$ is in \indep{\trt,k} iff the transducer $(\trt\downarrow\auta\uparrow\auta)$ is $(k+1)$-valued.
However, this is incorrect. We fix this as follows: Let $\trt^{=}$ be the transducer resulting from \trt if we add a new state $q$ that is both initial and final and has the loop transitions $(q,\sigma/\sigma,q)$ for all alphabet symbols $\sigma\in\Sigma$. Then $\tr^{=}(w)=\{w\}\cup\trt(w)$, for all words $w$. 
Let $\trs=(\trt^{=}\downarrow\auta\uparrow\auta)$. 
One verifies that
\begin{itemize}
    \setlength{\itemsep}{0pt}%
    \setlength{\parskip}{0pt}%
    \vspace{-0.5\topsep}
    \item $\indepP_{\trt,k}=\indepP_{\trt^{=},k}$.
	\item $\trs(u)=\trt^{=}(u)\cap L(\auta)$, for all $u\in L(\auta)$.
	\item $L(\auta)\in \indep{\trt,k}$ iff the transducer $\trs$ is $(k+1)$-valued.
\end{itemize}
The statement follows form the above mentioned result of \cite{GurIba:1983}.
\end{proof}

\pnsi
We note that, in \cite[Prop.~35]{IbaMcQu:2023}, the authors show that, for fixed $p,r,c,k$ there is a polynomial algorithm to decide whether a given ``2-way machine having $p$ counters with at most $r$ reversals and at most $c$ crossings on the input tape'' accepts a $k$-infix-free language (also,  $k$-prefix-free, $k$-suffix-free).
This implies that, for fixed $k$, the problem is decidable in polynomial time when the given machine is a standard automaton. 
In the above theorem we showed that the problem remains decidable in polynomial time for any marginal property in~$\indepP_{\trt,k}$.
\pnsi
In the above theorem, if $k$ is not fixed then the problem is decidable, but it is unlikely that this is so in polynomial time---as stated in Section~\ref{sec:dec}, already the problem is PSPACE-complete for the fixed property $\indepP_{\tsx,k}$.  
\pnsi
We now extend the definitions of the finitely-marginal properties of \cite{KoHanSalomaa:2021}, as we did in~\eqref{eq:TRANSD:margin} for the $k$-marginal properties:
\begin{equation}\label{eq:TRANSD:fin:margin}
L\in \indep{\trt,{\rm fin}}, \quad\text{if }\;
\trt(u)\cap L\; \text{ is finite,\, for all $u\in L$}.
\end{equation}
Again, \trt is any transducer (e.g., if $\trt=\tsx$ then \indep{\tsx,{\rm fin}} is the class of finitely-suffix-free languages).
As stated in Section~\ref{sec:dec}, the authors of \cite{KoHanSalomaa:2021} showed that the satisfaction problem for the  three marginal properties finitely prefix-free, finitely suffix-free and finitely infix-free is decidable in quadratic time for regular languages.
We show below that the \emph{uniform} satisfaction problem is also decidable in quadratic time w.r.t. the size of the given automaton.
The proof unifies the ideas in the proofs in~\cite{KoHanSalomaa:2021} in a way that works for all transducers \trt. 
For a transducer \trs, the proof uses the following terms: 
\begin{itemize}
    \setlength{\itemsep}{0pt}%
    \setlength{\parskip}{0pt}%
    \vspace{-0.5\topsep}
    \item \emdef{\ew-input path}: any path in \trs whose transition labels are of the form $\ew/b$ for $b\in\al\cup\{\ew\}$.
    \item \emdef{proper \ew-input cycle}: any cycle in \trs whose transition labels are of the form $\ew/b$, for $b\in\al\cup\{\ew\}$, and at least one of these labels has $b\not=\ew$.
    \item \emdef{trim part of \trs}: the transducer resulting if we eliminate the states of \trs that are not reachable from an initial state and cannot reach a final state.
\end{itemize}

\begin{theorem}\label{th:gen:sat:finite:margin}
Assuming a fixed (but arbitrary) alphabet, there is a $O(|\trt||\auta|^2)$ time algorithm to decide whether $L(\auta)\in\indepP_{\trt,{\rm fin}}$, for given transducer \trt and automaton \auta.
\end{theorem}
\begin{proof}
Let $\Sigma$ be the alphabet and let $\Sigma_{\ew}=\Sigma\cup\{\ew\}$.
We claim that the condition to test in statement~\eqref{eq:TRANSD:fin:margin} is equivalent to 
\[
\trs \; \text{ has no proper $\ew$-input cycle, where } \trs= \text{ the trim part of } \trt\downarrow\auta\uparrow\auta.
\]
See the proof of Theorem~\ref{th:gen:sat:margin} for the operations `$\downarrow$' and `$\uparrow$'. See further below for the correctness of the above claim. The time complexity follows when we note that 
(i) \trs is of size $O(|\trt||\auta|^2)$ and can be constructed in time $O(|\trt||\auta|^2)$; 
(ii) one can test in linear time whether a graph has a cycle---using breadth first search;
(iii) the graph to test for a cycle is the transducer that results when we treat \trs as an automaton (viewing the transition labels as atomic symbols) and intersect it with any fixed  automaton accepting the language $(\{\ew\}\times\Sigma_{\ew})^*(\{\ew\}\times\Sigma)(\{\ew\}\times\Sigma_{\ew})^*$.
\pnsi
We now establish the correctness of the claim about the transducer \trs. 
First assume that \trs has a proper \ew-input cycle with some label $(\ew,z)\not=(\ew,\ew)$, starting from some state $q$ and ending in $q$. As \trs is trim, there is a path from an initial state to $q$ with some label $(x,y)$, and a path from $q$ to a final state with some label $(x',y')$. 
Then $\trs(xx')$ is infinite, as it contains $yz^*y'$. Thus, for $u=xx'$, we have that $u\in L(\auta)$ and $\trt(u)\cap L(\auta)=\trs(u)$ is infinite and, therefore, $L(\auta)\notin\indepP_{\trt,{\rm fin}}$.
\pnsi
Now assume that $L(\auta)\notin\indepP_{\trt,{\rm fin}}$; then there is $u\in L(\auta)$ such that $\trs(u)$ is infinite. 
Consider the set $\pi_u$ of all accepting paths in \trs that contain \emph{no} $(\ew,\ew)$-cycles and have a label of the form $(u,v)$.
As $\trs(u)$ is infinite, the set $\pi_u$ is infinite.
As the input part $u$ of the labels is the same in all the $\pi_u$ paths, there are arbitrarily long \ew-input subpaths contained in the paths of $\pi_u$.
So at least one of those \ew-input subpaths has a repeated state and, therefore, the subpath has a cycle, and this cycle is not a $(\ew,\ew)$-cycle; hence it is a proper \ew-input cycle.
\end{proof}


\section{Concluding Remarks}\label{sec:last}
The concepts of $k$-prefix-free, $k$-suffix-free, and $k$-infix-free languages, and the associated results in \cite{KoHanSalomaa:2021} provide important insights that lead to a formal method for expressing the marginal properties $\indepP_{\trt,k}$ and studying the associated uniform satisfaction and maximality problems.
In particular, we have the following lines of research.
\begin{itemize}
    \setlength{\itemsep}{0pt}%
    \setlength{\parskip}{0pt}%
    \vspace{-0.5\topsep}
	\item 
	The condition to test for deciding whether a  language $L\in\indep{\trt}$ is maximal, where \trt is a transducer, \cite{DudKon:2012}, is
	  \[
	  L\cup\trt(L)\cup\trt^{-1}(L)=\alstar.
	  \]
	  Is there a similar condition to test for the \indep{\trt,k}-maximality of a  language $L$?
	\item 
	The PSPACE-completeness result in Section~\ref{sec:dec} about the satisfaction of the $k$-prefix-free property implies that the \emph{uniform} satisfaction problem of whether a regular language is in $\indepP_{\trt,k}$, for given \trt and $k$, is PSPACE-hard.
	Is the uniform satisfaction problem in PSPACE? If not, for which types of transducers \trt is the satisfaction problem in PSPACE---see again the discussion that follows Theorem~\ref{th:gen:sat:margin}?
\end{itemize}


\begin{credits}
\subsubsection{\ackname} This study was funded by a Discovery Grant of NSERC Canada (grant number RGPIN-2020-05996).

\subsubsection{\discintname}
The author has no competing interests to declare that are
relevant to the content of this article. 
\end{credits}

\if\MINE1\bibliographystyle{plain}
\else \bibliographystyle{splncs04}  
\fi
\bibliography{refs}

\end{document}